\documentclass{amsart}
\usepackage{graphicx}
\usepackage{graphicx}
\usepackage{amssymb}
\usepackage{amsfonts}
\usepackage{amsmath}
\usepackage{amsthm}
\newtheorem{thm}{Theorem}[section]
\newtheorem{defn}[thm]{Definition}
\newtheorem{lem}[thm]{Lemma}
\newtheorem{rem}[thm]{Remark}

\numberwithin{equation}{section}

\begin{document}

\title{Shortfall Risk Approximations for American Options in the multidimensional Black--Scholes Model}
 \author{Yan Dolinsky\\
 Department of Mathematics\\
 ETH, Zurich\\
 Switzerland }%

\address{
 Department of Mathematics, ETH, Zurich 8092, Switzerland\\
 {e.mail: yan.dolinsky@math.ethz.ch}}

 \date{\today}
\begin{abstract}
We show that shortfall
risks of American options in a sequence of multinomial
approximations of the multidimensional Black--Scholes (BS) market
converge to the corresponding quantities for similar American
options in the multidimensional BS market with path dependent
payoffs. In comparison to previous papers we consider the multi assets case
for which we use the weak convergence approach.
\end{abstract}

\subjclass[2000]{Primary: 91B28 Secondary: 60F15, 91B30}%
\keywords{American options, Shortfall risk, Weak convergence.}%
\maketitle
\markboth{Y.Dolinsly}{Shortfall risk Approximations}
\renewcommand{\theequation}{\arabic{section}.\arabic{equation}}
\pagenumbering{arabic}

\section{Introduction}\label{sec:1}\setcounter{equation}{0}
This paper deals with multinomial approximations of the shortfall risk for American options in the
multidimensional BS (complete) model.
It is well known that in a complete market an American contingent claim can
be hedged perfectly with an initial capital which is equal to the optimal stopping value of the discounted payoff under the unique
martingale measure. In real
market conditions an investor (seller) may not be willing for
various reasons to tie in a hedging portfolio the full initial
capital required for a perfect hedge. In this case the seller is
ready to accept a risk that his portfolio value at an exercise time
may be less than his obligation to pay and he will need additional
funds to fullfil the contract. Thus a portfolio shortfall comes into
the picture.

We deal with a certain type of risk called the
shortfall risk which is defined as the maximal expectation
(with respect to the buyer exercise times) of the discounted shortfall (see \cite{M}).
An investor whose initial capital is less than the option
price still want to compute the minimal possible shortfall risk
and to find a portfolio strategy which minimizes or "almost"
minimizes the shortfall risk. In this paper we
allow only \textit{admissible} self financing portfolios, i.e.
a portfolios with nonnegative wealth process.
This corresponds to the situation when the
portfolio is handled without borrowing of the capital.

For discrete time markets such as the multinomial
models the above problems can be solved by dynamical programming algorithm.
For continuous time models such as the BS model these
problems are much more complicated.

We prove that for American options, the shortfall risk in the multidimensional BS model can be
approximated by a sequence of shortfall risks
in an appropriate multinomial models.
This type of results has a practical value since
the shortfall risks in the multinomial models can be calculated via dynamical programming algorithm.
Our main tools are the extended weak convergence
theory that was developed in \cite{A2}
and the tightness theorems that were
obtained in \cite{MZ}. Since we use the weak convergence
approach we could not provide
error estimates of the above approximations.
Thus, to open problems remains open.
The first one is to obtain error estimates of the above approximations. The second one is to find explicit formulas
for optimal or "almost" optimal hedges in the BS model. It seems that both of the above problems require
new tools.

So far, shortfall risk approximations
were studied only in the one dimensional BS model (see \cite{DK2},
\cite{DK3}). For this case it was
proved that the shortfall risk in a BS market is a limit of the
shortfall risks in an
appropriate sequence of CRR markets.
Furthermore, the authors obtained error estimates
and dynamical programming algorithm for
"almost" optimal hedges. The main tool that
was used in the above papers is Skorohod
embedding tool of i.i.d. random variables
into the one dimensional Brownian motion.
This tool can not be applied for the multidimensional Brownian motion.

Main results of this paper are formulated in the next section. In Section 3 we
derive auxiliary lemmas that will be essential in the proof of the main results. In Section \ref{sec4} we
complete the proof of main results of the paper. In Section 5
we analyze the multinomial models and
provide a dynamical programming algorithm for the shortfall risk
and the corresponding optimal portfolios.

\section{Preliminaries and main results}\label{sec:2}\setcounter{equation}{0}
First we introduce the multidimensional BS market.
Consider a
complete probability space $(\Omega_W, P^W)$
together with a standard $d$--dimensional continuous in time
Brownian motion $\{W(t)=(W_1(t),...,W_d(t))\}_{t=0}^\infty$, and the filtration $\mathcal{F}^W_t=\sigma\{W(s)|s\leq t\}$.
We assume that the $\sigma$--algebras contain the null sets.
A BS
financial market consists of a savings account $B(t)$ with an
interest rate $r$, assuming without loss of generality that $r=0$,
i.e.
\begin{equation}\label{2.1}
B(t)=B(0)>0
\end{equation}
and of $d$ risky stocks $S^{W}=(S^{W}_1,...,S^{W}_d)$ given by the following
equation
\begin{equation}\label{2.2}
\begin{split}
S^{W}_i(t)=S_i(0)\exp(\sum_{j=1}^d
\sigma_{ij}W_j(t)+(b_i-\frac{1}{2} \sum_{j=1}^d \sigma^2_{ij})t),
\ S_i(0)>0
\end{split}
\end{equation}
where $b\in\mathbb{R}^d$ is a constant vector and $\sigma\in
M_d(\mathbb{R})$ is a constant nonsingular matrix.

Let $T<\infty$ be the
maturity date of our American option and let $\mathcal{T}^W_{[0,T]}$ be the set of all stopping
times with respect to $\mathcal{F}^W$ which take values in $[0,T]$.
Denote by
$(\mathbb{D}([0,T];\mathbb{R}^d),\mathcal{S})$ the space
of all right continuous functions with left hand limits, equipped with the Skorohod topology (see \cite{B}).
Let $F:[0,T]\times(\mathbb{D}([0,T];\mathbb{R}^d),\mathcal{S})\rightarrow\mathbb{R}_{+}$
be a measurable functions such that there exists a constant $C>0$ which satisfies
\begin{equation}\label{2.3}
\begin{split}
\sup_{0\leq t\leq T} F(t,x)\leq C\sup_{0\leq t \leq T} |x(t)|, \ \ \forall{x}\in\mathbb{D}([0,T];\mathbb{R}^d).
\end{split}
\end{equation}
Furthermore, we assume that
for any $t\in [0,T]$ and $x,y\in \mathbb{D}([0,T];\mathbb{R}^d)$:\\
i. $F(\cdot,x)$ is a right continuous function with left hand
limits.\\
ii. $F(t,x)=F(t,y)$ if $x(s)=y(s)$ for any $s\leq{t}$.\\
iii. If $x$ is continuous at $t$ then $F$ is continuous at $(x,t)$ (with respect to the product topology).

Next, consider an American option with the
payoff process given by
\begin{equation}\label{2.4}
Y^W(t)=F(t,S^W), \ 0\leq t\leq{T}.
\end{equation}
From the assumptions above it follows that ${\{Y^W(t)\}}_{t=0}^T$ is a  $c\grave{a}dl\grave{a}g$
adapted stochastic process and $E^W[\sup_{0\leq t\leq T}Y^W(t)]<\infty$. Denote by $\tilde{P}^W$ the unique martingale measure for the above
model. Using standard arguments it follows that the restriction of
the probability measure $\tilde{P}^{W}$ to the $\sigma$--algebra
$\mathcal{F}^{W}_t$ satisfies
\begin{equation}\label{2.5}
M(t)=\frac{d\tilde{P}^W}{dP^{W}}|\mathcal{F}^{W}_t=\exp(-\frac{1}{2}||\theta||^2t-
\langle\theta,W(t)\rangle)
\end{equation}
where $\theta=b\sigma^{*}$. We denote by $||\cdot||$ and $\langle\cdot,\cdot\rangle$
the standard norm and the scalar product of $\mathbb{R}^d$, respectively.

A self financing strategy
$\pi$ with a horizon $T$ and an initial capital $x$ (see \cite{S}) is a
$d$--dimensional progressively measurable process
$\pi=\{\gamma(t)\}_{t=0}^{T}$ which satisfies
\begin{equation}\label{2.6}
\begin{split}
\int_{0}^{T}{\langle\gamma(t),{S}^W(t)\rangle}^2dt<\infty\ \ \mbox{a.s.}
\end{split}
\end{equation}
For a strategy $\pi$ the portfolio value process
${\{V^{\pi}(t)\}}_{t=0}^T$ is given by
\begin{equation}\label{2.7}
{V}^{\pi}(t)=x+\int_{0}^{t} \langle\gamma(u),d{{S}^W(u)}\rangle.
\end{equation}
Recall, (see \cite{LS}) that stochastic integrals with respect to the Brownian motion has a continuous modification
and so for any self financing strategy $\pi$ the corresponding portfolio value process is a continuous one.

A self financing strategy $\pi$ is called \textit{admissible} if
${V}^{\pi}(t)\geq{0}$ for all $t\in{[0,{T}]}$ and the set of such
strategies with an initial capital no bigger than $x$ will be denoted by
$\mathcal{A}^W(x)$. We set $A^W=\bigcup_{x>0}A^W(x)$.
For an \textit{admissible} self financing strategy
$\pi$ the shortfall risk is given by (see \cite{M}),
\begin{equation}\label{2.8}
R(\pi)=\sup_{\tau\in\mathcal{T}^W_{[0,T]}}E^W[(Y^W({\tau})-
{V}^{\pi}(\tau))^+],
\end{equation}
which is the maximal possible expectation with respect to the
probability measure $P^W$ of the (discounted) shortfall.
The shortfall risk for
an initial capital $x$ is given by
\begin{equation}\label{2.9}
R(x)=\inf_{\pi\in\mathcal{A}^W(x)} R(\pi).
\end{equation}

Next, we introduce the sequence of multinomial markets that we use
in order to approximate the shortfall risk in the BS model. The same
markets were used in \cite{H} in order to approximate European
option prices in the $d$-- dimensional BS model. Let $A\in
M_{d+1}(\mathbb{R})$ be an orthogonal matrix such that it last
column equals to $(\frac{1}{\sqrt{d+1}},...,\frac{1}{\sqrt{d+1}})^{*}$.
Let $\Omega_{\xi}={\{1,2,...,d+1\}}^\infty$ be the space of finite
sequences $\omega=(\omega_1,\omega_2,...)$;
$\omega_i\in{\{1,2,...,d+1\}}$ with the product probability
$P^{\xi}={\{\frac{1}{d+1},...,\frac{1}{d+1}\}}^\infty$. Define a
sequence of i.i.d. random vectors $\xi^{(1)},\xi^{(2)},...$ by
\begin{equation}\label{2.10}
\begin{split}
\xi^{(i)}(\omega)=\sqrt{d+1}(A_{\omega_i 1},A_{\omega_i
2}...,A_{\omega_i d}), \ \ i\in\mathbb{N}.
\end{split}
\end{equation}
Let $\mathcal{F}^{\xi}_m=\sigma\{\xi^{(k)}|k\leq{m}\}$, $m\geq 0$ $(\mathcal{F}^{\xi}_0=\{\emptyset,\Omega_{\xi}\}$).
Denote by $\mathcal{T}^{\xi}_m$ the set of all stopping times with
respect to the filtration ${\{\mathcal{F}^{\xi}_k\}}_{k=0}^\infty$
with values from $0$ to $m$.

For any $n$ consider the $n$--step multinomial market which consists
of a savings account $B^{(n)}(t)$ given by
\begin{equation}\label{2.11}
B^{(n)}(t)=B(0)>0
\end{equation}
and of $d$ risky stocks $S^{\xi,n}=(S^{\xi,n}_1,...,S^{\xi,n}_d)$ given by the
formulas $S^{\xi,n}_i(t)=S_i(0)$ for $t\in{[0,T/n)}$ and
\begin{equation}\label{2.12}
\begin{split}
S^{\xi,n}_i(t)=S_i(0)\prod_{m=1}^{k}(1+\frac{b_i
T}{n}+\sqrt{\frac{T}{n}}\sum_{j=1}^{d}\sigma_{ij}\xi^{(m)}_j),  \
kT/n \leq t<(k+1)T/n, \ k=1,...,n.
\end{split}
\end{equation}
We assume that $n$ is sufficiently large such that the terms in the
above product are positive a.s. The market
is active at the times $0,\frac{T}{n},\frac{2T}{n},...,T$. It is well known that this
market is complete and we denote by $\tilde{P}^{\xi}_n$ the unique
martingale measure. Define the stochastic process
${\{M^{(n)}(t)\}}_{t=0}^T$
\begin{equation}\label{2.13}
\begin{split}
M^{(n)}(t)=\frac{d\tilde{P}^{\xi}_n}{dP^{\xi}}|\mathcal{F}^{\xi}_k, \
\ kT/n \leq t<(k+1)T/n, \ k=0,1,...,n.
\end{split}
\end{equation}
Clearly ${\{M^{(n)}(\frac{kT}{n})\}}_{k=0}^n$ is a martingale with respect
to the probability measure $P^\xi$ and the filtration
${\{\mathcal{F}^{\xi}_k\}}_{k=0}^n$. Explicit formulas for
$M^{(n)}(t)$ were obtained in \cite{H}. Consider an American option
with the adapted payoff process
\begin{equation}\label{2.14}
Y^{\xi,n}(k)=F(\frac{kT}{n},S^{\xi,n}), \  \ 0 \leq k\leq n.
\end{equation}
A self financing strategy $\pi$ with an initial capital $x$ and a
horizon $n$ (see \cite{S}) is a sequence $\pi=(\gamma(1),...,\gamma(n))$ where
$\gamma(k)$ are $\mathcal{F}^{\xi}_{k-1}$-measurable random vectors.
The portfolio value $V^\pi(k)$, $k=0,1,...,n$ is given by
\begin{equation}\label{2.15}
V^\pi(k)=x+\sum_{i=0}^{k-1}\langle\gamma({i+1}),
({S}^{\xi,n}((i+1)T/n)-{S}^{\xi,n}(iT/n))\rangle.
\end{equation}
We call a self financing strategy $\pi$ \textit{admissible} if
$V^\pi(k)\geq{0}$ for any $k\leq{n}$. Denote by
$\mathcal{A}^{\xi,n}(x)$ the set of all \textit{admissible} self
financing strategies with an initial capital no bigger than $x$, let $\mathcal{A}^{\xi,n}=
\bigcup_{x>0}\mathcal{A}^{\xi,n}(x)$. The definitions
for the shortfall risks in the multinomial markets are similar to
the definitions in the BS model. Thus for the $n$--step multinomial
market the shortfall risks are given by
\begin{equation}\label{2.16}
\begin{split}
R_n(\pi)=\max_{\tau\in\mathcal{T}^{\xi}_n}E^\xi[(Y^{\xi,n}({\tau})-V^{\pi}({\tau}))^{+}]
\ \ \mbox{and} \ \ R_n(x)=
\inf_{\pi\in\mathcal{A}^{\xi,n}(x)}R_n(\pi),
\end{split}
\end{equation}
where $E^\xi$ is the expectation with respect to $P^\xi$.

The following theorem is the main result of the paper and it says that the shortfall risk $R(x)$ for an
initial capital $x$ of an American option in the multidimensional
BS market can be approximated by a sequence of shortfall risks with
an initial capital $x$ of an American options in the multinomial
markets defined above. This result has a practical value since for any $n$ the shortfall risk $R_n(x)$
can be calculated by dynamical programming algorithm which is given in Section 5.
\begin{thm}\label{thm2.1}
For any $x>0$
\begin{equation}\label{2.17}
lim_{n\rightarrow\infty}R_n(x)=R(x).
\end{equation}
\end{thm}

The proof (which is given in Section 4) consists of two parts. In the first part
we prove the inequality
$R(x)\leq \liminf_{n\rightarrow\infty}R_n(x)$ and in the second part we prove that $R(x)\geq \limsup_{n\rightarrow\infty}R_n(x)$.
In the first part we take a sequence of "almost" optimal portfolios ${\{\pi_n\}}_{n=1}^\infty$
for the multinomial markets and consider their limit in some sense that will be explained explicitly in Section 3. From the limit process we construct
a portfolio $\pi$ in the BS model such that $R(\pi)\leq \liminf_{n\rightarrow\infty}R_n(\pi_n)=\liminf_{n\rightarrow\infty}R_n(x)$. The second part is proved
by a reversed operations. Namely, we take an "almost" optimal portfolio $\pi$ in the BS model which has some smoothness properties. The existence of
such portfolio will be proved by applying density arguments.
From this portfolio we construct a sequence of portfolios ${\{\pi'_n\}}_{n=1}^\infty$ in the multinomial models which satisfy
$\limsup_{n\rightarrow\infty}R_n(\pi'_n)\leq R(\pi)$.

\section{Auxiliary lemmas}\label{sec3}\setcounter{equation}{0}
Let $I\subset [0,T]$ be a dense set in $[0,T]$ and let $\mathcal{T}_I\subset\mathcal{T}^W_{[0,T]}$
be the set of all stopping times with a finite number of values which belongs to $I$.
\begin{lem}\label{lem3.2}
For any $\pi\in\mathcal{A}^W$,
\begin{equation}\label{3.2}
R(\pi)=\sup_{\tau\in \mathcal{T}_I}E[(Y^W({\tau})-V^\pi({\tau}))^{+}].
\end{equation}
\end{lem}
\begin{proof}
Choose $\epsilon>0$. There exists $\tau\in\mathcal{T}^W_{[0,T]}$ such that
\begin{equation}\label{3.3}
R(\pi)< E[(Y^W({\tau})-V^\pi({\tau}))^{+}]+\epsilon.
\end{equation}
For any $n$ there exists a finite set $I_n\subset I$ for which $\bigcup_{z\in I_n} (z-\frac{1}{n},z+\frac{1}{n})\supseteq [0,T]$.
Let $a_n$ be the maximal element of $I_n$. Define $\tau_n=\min\{t\in I_n|t\geq\tau\}\mathbb{I}_{\tau_n\leq a_n}+a_n \mathbb{I}_{\tau_n>a_n}$,
 where
$\mathbb{I}_D=1$ if an event $D$ occurs and =0 if not.
Clearly, $\tau_n\leq a_n$ a.s. and for $t\in I_n\setminus\{a_n\}$ we have $\{\tau_n\leq t\}=\{\tau\leq t\}\in \mathcal{F}^W_t$.
Thus $\tau_n\in \mathcal{T}_I$. Furthermore, $|\tau_n-\tau|\leq \frac{2}{n}$ and so $\tau_n\rightarrow\tau$ a.s.
From (\ref{3.3}) and the assumptions on $F$ we obtain
\begin{eqnarray}\label{3.4}
&R(\pi)< \epsilon+E[\lim_{n\rightarrow\infty}(Y^W(\tau_n)-V^\pi({\tau_n}))^{+}]=\epsilon+\\
&\lim_{n\rightarrow\infty}E[(Y^W(\tau_n)-V^\pi(\tau_n))^{+}]\leq
\epsilon+\sup_{\tau\in \mathcal{T}_I}E[(Y^W({\tau})-V^\pi({\tau}))^{+}]\nonumber
\end{eqnarray}
and the result follows by letting $\epsilon\downarrow{0}$.
\end{proof}

The next lemma provides a general result for the shortfall risk measure.
\begin{lem}\label{lem3.3-}
Let $x>0$. For any $\epsilon>0$ there exists $\psi\in
C((\mathbb{D}([0,T];\mathbb{R}^d),\mathcal{S}))$ such that the
martingale which given by
$Q(t)=E^W(\psi(S^W)|\mathcal{F}^W_t)$, $t\leq{T}$ is satisfying
\begin{equation}\label{3.15+}
\begin{split}
Q(0)<x  \ \ \mbox{and} \ \
R(x)>\sup_{\tau\in\mathcal{T}^W_{[0,T]}}E^W \bigg(\bigg(Y^W({\tau})-\frac{Q({\tau})}{M({\tau})}\bigg)^{+}\bigg)-\epsilon.
\end{split}
\end{equation}
\end{lem}
\begin{proof}
Let $\epsilon>0$. Set
$K=E^W[\sup_{0\leq t\leq T}\frac{1}{M({t})}]<\infty$
and $\delta=\frac{\epsilon}{2(K+1)}$. There exists
$\pi\in\mathcal{A}(x)$ such that $R(\pi)<R(x)+\delta$. The process
$\Phi(t):=V^{\pi}(t)M(t)$, $t\leq{T}$ is a
supermartingale with respect to $P^W$. Introduce the regular martingale
$\Gamma(t)=E(\sup_{0\leq u\leq T}Y^W(u)M(u)|\mathcal{F}^W_t)$, $t\leq T$.
The process $\Psi(t):=\Phi(t)\wedge \Gamma(t)$ is a
supermartingale of class $\mathcal{D}$. By
Doob's decomposition theorem there exists a continuous martingale
${\{U(t)\}}_{t=0}^T$ such that $U(0)=\Psi(0)\leq\Phi(0)=x$ and
$U(t)\geq\Psi(t)$ a.s. for all $t\leq{T}$. Observe that
\begin{eqnarray}\label{3.15++}
&\sup_{\tau\in\mathcal{T}^W_{[0,T]}}E^W\bigg(\bigg(Y^W({\tau})-\frac{U({\tau})}{M({\tau})}\bigg)^{+}\bigg)\leq
\sup_{\tau\in\mathcal{T}^W_{[0,T]}}E^W\bigg(\bigg(Y^W({\tau})-\\
&\frac{\Psi({\tau})}{M({\tau})}\bigg)^{+}\bigg)=
\sup_{\tau\in\mathcal{T}^W_{[0,T]}}E^W[(Y^W({\tau})-V^{\pi}({\tau}))^{+}]<R(x)+\delta.\nonumber
\end{eqnarray}
Next, choose a sequence $0\leq\psi_n\in
C((\mathbb{D}([0,T];\mathbb{R}^d),\mathcal{S}))$, $n\geq{1}$ such that
\begin{equation}\label{3.15+++}
\begin{split}
\lim_{n\rightarrow\infty}E^W|\psi_n(S^W)-U(T)|=0 \ \mbox{and} \
E^W\psi_n(S^W)<E^W U(T)\leq{x}, \ n\in\mathbb{N}.
\end{split}
\end{equation}
Set $Q^{(n)}(t)=E^W(\psi_n(S^W)|\mathcal{F}^W_t)$, $t\leq{T}$ and
introduce the set $C_n=\{\sup_{0\leq t\leq T}|U(t)$
$-Q^{(n)}(t)|>\delta\}$. From (\ref{3.15++}) we obtain that for any $n$,
\begin{eqnarray}\label{3.15++++}
&\sup_{\tau\in\mathcal{T}^W_{[0,T]}}
E^W\bigg(\bigg(Y^W({\tau})-\frac{Q^{(n)}({\tau})}{M({\tau})}\bigg)^{+}\bigg)\leq
\sup_{\tau\in\mathcal{T}^W_{[0,T]}}
E^W\bigg(\bigg(Y^W({\tau})-\\
&\frac{U({\tau})}{M({\tau})}\bigg)^{+}\bigg)
+\delta E^W[\sup_{0\leq t\leq T}\frac{1}{M({t})}]+E^W(\mathbb{I}_{C_n}\sup_{0\leq t\leq T}Y^W(t))\nonumber\\
&<R(x)+\frac{\epsilon}{2}+E^W(\mathbb{I}_{C_n}\sup_{0\leq t\leq T}Y^W(t)).\nonumber
\end{eqnarray}
By using the Doob inequality for the
continuous submartingale ${\{|U(t)-Q^{(n)}(t)|\}}_{t=0}^T$, it follows from (\ref{3.15+++})
that $\lim_{n\rightarrow\infty}P(C_n)=0$. This together with (\ref{3.15++++})
gives that for sufficiently large $n$, $\sup_{\tau\in\mathcal{T}^W_{[0,T]}}
E^W\bigg(\bigg(Y^W({\tau})-\frac{Q^{(n)}({\tau})}{M({\tau})}\bigg)^{+}\bigg)<R(x)+\epsilon$,
as required.
\end{proof}

Given a probability space
$(\Omega,\mathcal{F},P)$ consider a $c\grave{a}dl\grave{a}g$
stochastic process
$S=\{S_t:\Omega$\\
$\rightarrow{\mathbb{R}^d}\}_{t=0}^\Theta$, ($\Theta<\infty$). Denote by
$\mathcal{F}^S={\{\mathcal{F}^S_t\}}_{t=0}^\Theta$ the usual filtration
of $S$ i.e. the smallest right continuous filtration with respect to
which $S$ is adapted, and such that the $\sigma$--algebras contain the
null sets. Let $\mathcal{T}^S_{[0,\Theta]}$ be the set of all stopping
times with respect to $\mathcal{F}^S$ which take values in $[0,\Theta]$.

In \cite{MZ} the authors introduced the Meyer--Zheng (MZ) topology on the space $\mathbb{D}([0,\Theta]$\\$;\mathbb{R})$.
This topology will denoted by
$(\mathbb{D}([0,\Theta];\mathbb{R}),MZ)$. The MZ
topology is in fact the topology of convergence in measure, it is
weaker than the Skorohod topology, but for the MZ topology any
sequence of positive uniformly $L^1$--bounded supermartingales is
relatively compact (see \cite{MZ}). This fact together with the following lemma will be essential in the proof of Theorem \ref{thm2.1}.
\begin{lem}\label{lem3.3}
Let $(\Omega,\mathcal{F},P)$ be a probability space and
$S^{(n)}:\Omega\rightarrow\mathbb({D}[0,\Theta];$ $\mathbb{R}^d)$ be a
sequence of stochastic processes such that $S^{(n)}\rightarrow{S}$ a.s. on the space
$(\mathbb{D}([0,\Theta];\mathbb{R}^d),\mathcal{S})$. Assume that for any $n$,
${\{V^{(n)}(t)\}}_{t=0}^\Theta$ is a (one dimensional)
$c\grave{a}dl\grave{a}g$ positive supermartingale with respect to
the filtration $\mathcal{F}^{S_n}_{[0,\Theta]}$ and
$V^{(n)}\rightarrow{V}$ a.s. on the space $(\mathbb{D}([0,\Theta];\mathbb{R}),MZ)$ with respect to the MZ topology. Set
\begin{equation}\label{3.16}
Q(t)=E(V(t)|\mathcal{F}^S_t), \ \ t\leq{\Theta}.
\end{equation}
Then the process ${\{Q(t)\}}_{0\leq t< \Theta}$ is a $c\grave{a}dl\grave{a}g$ positive
supermartingale with respect to the filtration
$\mathcal{F}^{S}.$.
\end{lem}
\begin{proof}
First, let us show that  ${\{Q(t)\}}_{0\leq t< \Theta}$ is a supermartingale,
i.e. for any $s<t<\Theta$ and $D\in\mathcal{F}^S_s$
\begin{equation}\label{3.16+}
E\mathbb{I}_DV(s)\geq E\mathbb{I}_DV(t).
\end{equation}
Choose $s<s'<t$, $c>0$ and $0<\epsilon<\min(s'-s,\Theta-t)$.  Let
$\phi\in C((\mathbb{D}([0,\Theta];\mathbb{R}^d),\mathcal{S}))$ be a continuous bounded
function such that $\phi(x)$ depends only on the restriction of $x$
to the interval $[0,s']$. From the definition of the MZ topology we
obtain
\begin{eqnarray*}\label{3.17}
&\limsup_{n\rightarrow\infty}E\int_{u=0}^{\epsilon}
|\phi(S^{(n)})(V^{(n)}({s'+u})\wedge{c})-\phi(S)(V({s'+u})\wedge{c})|du\\
&\leq||\phi||_{\infty}\limsup_{n\rightarrow\infty}E\int_{u=0}^{\epsilon}
(|V^{(n)}({s'+u})-V({s'+u})|\wedge c)du+\nonumber\\
&c\limsup_{n\rightarrow\infty}
E\int_{u=0}^{\epsilon}|\phi(S)-\phi(S^{(n)})|du=0.\nonumber
\end{eqnarray*}
Thus,
\begin{equation}\label{3.18}
\lim_{n\rightarrow\infty}\frac{1}{\epsilon}E\int_{u=0}^{\epsilon}
\phi(S^{(n)})(V^{(n)}({s'+u})\wedge{c})du=\frac{1}{\epsilon}E\int_{u=0}^{\epsilon}
\phi(S)(V({s'+u})\wedge{c})du.
\end{equation}
Similarly,
\begin{equation}\label{3.19}
\lim_{n\rightarrow\infty}\frac{1}{\epsilon}E\int_{u=0}^{\epsilon}
\phi(S^{(n)})(V^{(n)}({t+u})\wedge{c})du=\frac{1}{\epsilon}E\int_{u=0}^{\epsilon}
\phi(S)(V({t+u})\wedge{c})du.
\end{equation}
For any $n$, ${\{V^{(n)}({\alpha})\wedge{c}\}}_{\alpha=0}^U$ is a
supermartingale with respect to $\mathcal{F}^{S^{(n)}}_{[0,\Theta]}$,
this together with (\ref{3.18}) and (\ref{3.19}) gives
\begin{equation*}
\frac{1}{\epsilon}E\int_{u=0}^{\epsilon}
\phi(S)(V({t+u})\wedge{c})du\leq
\frac{1}{\epsilon}E\int_{u=0}^{\epsilon}
\phi(S)(V({s'+u})\wedge{c})du.
\end{equation*}
By taking $\epsilon\downarrow{0}$ we obtain
$E\phi(S)(V({t})\wedge{c})\leq E\phi(S)(V({s'})\wedge{c})$. From density arguments
and the fact that $D\in\sigma\{S_u|u<s'\}$
it follows that
$E\mathbb{I}_D(V(s')\wedge c)\geq E\mathbb{I}_D(V(t)\wedge c)$ and by letting
$s'\downarrow{s}$ and $c\uparrow \infty$ we obtain (\ref{3.16+}). Finally, since the map
$t\rightarrow{EQ(t)}=EV(t)$ is right continuous we obtain (see \cite{LS})
that $Q$ has a $c\grave{a}dl\grave{a}g$ modification.
\end{proof}

In \cite{H} it was proved that
\begin{equation}\label{3.20---}
\begin{split}
(S^{\xi,n},M^{(n)})\Rightarrow
(S^W,M) \ \mbox{on} \ \mbox{the} \ \mbox{space} \ (\mathbb{D}([0,T];\mathbb{R}^d),\mathcal{S})\times
(\mathbb{D}([0,T];\mathbb{R}),\mathcal{S}).
\end{split}
\end{equation}
We use the notation $S^{(n)}\Rightarrow S$ to
indicate that the sequence $S^{(n)}$, $n\geq 1$ converges
weakly to $S$ (see \cite{B}).
We will use the concept "extended weak
convergence" which was introduced in \cite{A2} by Aldous. The original definition
was via prediction processes.  For the case where the
stochastic processes are considered with respect to their usual
filtration he proved that extended weak
convergence is equivalent to a more elementary condition which does
not require the use of prediction processes (see \cite{A2} Proposition
16.15). We will use the above condition as the
definition of extended weak convergence.

\begin{defn}\label{defn3.1}
A sequence
$S^{(n)}:\Omega_n\rightarrow\mathbb{D}([0,T];\mathbb{R}^d)$,
$n\geq{1}$ extended weak converges to a
stochastic process
$S:\Omega\rightarrow\mathbb{D}([0,T];\mathbb{R}^d)$ if for any $k$
and continuous bounded functions
$\psi_1,...,\psi_k\in{C((\mathbb{D}([0,T];\mathbb{R}^d),\mathcal{S}))}$
\begin{equation}\label{3.20--}
(S^{(n)},H^{(n,1)},...,H^{(n,k)})\Rightarrow (S,H^{(1)},...,H^{(k)})
\ \mbox{on} \ (\mathbb{D}([0,T];\mathbb{R}^{d+k}),\mathcal{S})
\end{equation}
where for any $t\leq{T}$, $1\leq i\leq{k}$ and $n\in\mathbb{N}$
\begin{equation}\label{3.20-}
H^{(n,i)}_t=E_n(\psi_i(S^{(n)})|\mathcal{F}^{S^{(n)}}_t),
n\in\mathbb{N}, \ \mbox{and} \ H^{(i)}=E(\psi_i(S)|\mathcal{F}^S_t)
\end{equation}
$E_n$ denotes the expectation with respect to the probability
measure on $\Omega_n$ and $E$ denotes the expectation with respect
to the probability measure on $\Omega$. We will denote extended weak
convergence by $S^{(n)}\Rrightarrow{S}$.
\end{defn}
\begin{lem}\label{lem3.4}
$S^{\xi,n}\Rrightarrow S^{W}.$
\end{lem}
\begin{proof}
Define the map
$G:(\mathbb{D}([0,T];\mathbb{R}^d),\mathcal{S})\rightarrow
(\mathbb{D}([0,T];\mathbb{R}^d),\mathcal{S})$ by
$(G(x_1,...,x_d))(t)$\\$=(\exp(x_1(t),...,\exp(x_d(t))$. Observe that $G$ is a continuous
map with continuous inverse (the inverse is defined only on
functions $(x_1,...,x_d)\in \mathbb{D}([0,T];\mathbb{R}^d)$ which
satisfy $\min_{1\leq i\leq d}\inf_{0\leq{t}\leq{T}}x_i(t)>0$).
Let  ${\{X(t)=(\ln S^{W}_1(t),...,\ln
S^{W}_d(t))\}}_{t=0}^T$ and ${\{X^{(n)}(t)=(\ln
S^{\xi,n}_1(t),...,\ln S^{\xi,n}_d(t))\}}_{t=0}^T$, $n\in\mathbb{N}$.
From (\ref{3.20---}) and the fact that $G$ has a continuous inverse it follows that
$X^{(n)}\Rightarrow X$.
For any $n$ the process $X^{(n)}$
has an independent
increments and the process $X$ is a continuous process with independent
increments. From Corollary 2 in \cite{JS} we obtain $X^{(n)}\Rrightarrow X$ and so (since $G$ is continuous)
$S^{\xi,n}\Rrightarrow S^{W}.$
\end{proof}

\section{Proof of main results}\label{sec4}\setcounter{equation}{0}
In this section we complete the proof of Theorem
\ref{thm2.1}. Fix $x$. We start with the proof of the
inequality $R(x)\leq\lim_{n\rightarrow\infty}R_n(x)$. Here and in the
sequel, for the sake of simplicity we will assume that indices have
been renamed so that the whole sequence converges.
Let $\pi_n\in\mathcal{A}^{\xi,n}(x)$, $n\in\mathbb{N}$
be a sequence such that
\begin{equation}\label{4.12}
\begin{split}
R_n(\pi_n)<R_n(x)+\frac{1}{n} \ \ \forall{n\in\mathbb{N}}.
\end{split}
\end{equation}
For any $n\in\mathbb{N}$ define the stochastic process
${\{Z^{(n)}(t)\}}_{t=0}^{2T}$ by $Z^{(n)}(t)=V^{\pi_n}_kM^{(n)}({t})$ for
$\frac{kT}{n}\leq t<\frac{(k+1)T}{n}$ and $k<n$, and
$Z^{(n)}(t)=V^{\pi_n}(n)M^{(n)}({T})$ for $t\geq T$.
From (\ref{2.13}) it follows
that ${\{Z^{(n)}(t)\}}_{t=0}^{2T}$ is a $c\grave{a}dl\grave{a}g$
martingale with respect to $P^{\xi}$ and the filtration
${\{\mathcal{F}^{S^{n,\xi}}_t\}}_{t=0}^{2T}$, where we set
$\mathcal{F}^{S^{n,\xi}}_t=\mathcal{F}^{S^{n,\xi}}_T$ for $t\geq T$.
From \cite{MZ} it follows that the sequence $Z^{(n)}$,
$n\in\mathbb{N}$ in tight on the space $(\mathbb{D}([0,T];\mathbb{R}),MZ)$.
We can extend all the processes in (\ref{3.20---}) to the interval
$[0,2T]$ be letting their paths
to be constants on the interval $[T,2T]$. From (\ref{3.20---}) we obtain that the sequence
$(S^{\xi,n},M^{(n)},Z^{(n)})$, $n\in\mathbb{N}$
it tight on the space
$(\mathbb{D}([0,2T];\mathbb{R}^d),\mathcal{S})\times
(\mathbb{D}([0,2T];\mathbb{R}),\mathcal{S})\times (\mathbb{D}([0,2T];\mathbb{R}),MZ)$. Thus there exists a subsequence such
that $(S^{\xi,n},M^{(n)},Z^{(n)})\Rightarrow (S^W,M,Z)$,
for some stochastic process $Z$
which satisfies $Z(0)\leq{x}$. Next, from the Skorohod
representation theorem (see \cite{D1}) it follows that without loss of generality we
can assume that there exists a probability space
$(\Omega,\mathcal{F},P)$ on which
\begin{equation}\label{4.14}
(S^{\xi,n},M^{(n)},Z^{(n)})\rightarrow (S^W,M,Z) \ \mbox{ a.s.}
\end{equation}
on the space $(\mathbb{D}([0,2T];\mathbb{R}^d),\mathcal{S})\times
(\mathbb{D}([0,2T];\mathbb{R}),\mathcal{S})\times (\mathbb{D}([0,2T];\mathbb{R}),MZ)$.
From Lemma \ref{lem3.3} it follows that the process
$Q(t):=E(Z(t)|\mathcal{F}^{S^W}_t)$, $t\leq{T}$ is a
$c\grave{a}dl\grave{a}g$ supermartingale. The process $V(t):=\frac{Q(t)\wedge \Gamma(t)}{M(t)}$, $t\leq T$ is a
$c\grave{a}dl\grave{a}g$ supermartingale of class $\mathcal{D}$ with respect to the martingale measure $\tilde{P}^W$
($\Gamma(t)$ was introduced after (\ref{3.15+})).
From
Doob's decomposition theorem and the martingale representation theorem we obtain that there exists
a portfolio $\pi\in\mathcal{A}(x)$ such that
\begin{equation}\label{4.14+}
\begin{split}
V^\pi(0)=V(0)\leq Q(0)=Z(0)\leq x \ \ \mbox{and} \ \ V^\pi(t)\geq V(t) \ \forall{t\leq T}.
\end{split}
\end{equation}
From \cite{MZ} there exists a subsequence $Z^{(n)}$ and a dense set
$I\subset{[0,T]}$, such that for any $t\in J$
\begin{equation}\label{4.15}
\lim_{n\rightarrow\infty} Z^{(n)}(t)=Z(t) \ \  \mbox{a.s.}
\end{equation}
Choose $\epsilon>0$. From Lemma \ref{lem3.2} we obtain that there exist
a stopping time $\tau$ which excepts a finite number of values
$\{t_1<t_2<...<t_m\}\subset I$
 such that
\begin{equation}\label{4.15+}
R(\pi)<\epsilon+E[(Y^W(\tau)-V^\pi({\tau}))^{+}].
\end{equation}
From Lemma 3.2 in \cite{D} and (\ref{4.14}) it follows that there exists a sequence
$\sigma_n\in\mathcal{F}^{S^{\xi,n}}_{[0,T]}$, $n\geq{1}$ of stopping times with values in the set
$\{t_1<t_2<...<t_m\}$ which satisfy
\begin{equation}\label{4.15++}
\begin{split}
\lim_{n\rightarrow\infty}\sigma_n=\tau \ \ \mbox{a.s.}
\end{split}
\end{equation}
Set
$\tau_n=\max\{k|kT/n\leq\sigma_n\}$, $n\geq{1}$. Observe that for any $k\leq n$, $\{\tau_n\leq k\}=\{\sigma_n<(k+1)T/n\}\in\mathcal{T}^\xi_k$
thus for any $n$, $\tau_n\in\mathcal{T}^\xi_n$.
Furthermore,
\begin{equation}\label{4.15++}
\begin{split}
|\frac{\tau_nT}{n}-\sigma_n|\leq\frac{1}{n} \ \ \mbox{and} \ \ Z^{(n)}({\sigma_n})=Z^{(n)}(\tau_nT/n) \ \ \forall{n}.
\end{split}
\end{equation}
From (\ref{2.3}) it follows that the random variables
$Y^{\xi,n}({\tau_n})$, $n\in\mathbb{N}$
are uniformly integrable. Thus, from Jensen's inequality and (\ref{4.12})--(\ref{4.15++}) it follows
\begin{eqnarray}\label{4.17-}
&R(x)\leq R(\pi)\leq\epsilon+E[(Y^W(\tau)-V^\pi({\tau}))^{+}]\leq \epsilon+\\
&E\bigg(\bigg(Y^W(\tau)-\frac{Q({\tau})}{M(\tau)}\bigg)^{+}\bigg)
\leq\epsilon+E\bigg(E\bigg(Y^W(\tau)-\frac{Z({\tau})}{M(\tau)}\bigg)^{+}\bigg|\mathcal{F}^{S^W}_\tau\bigg)\bigg)=\nonumber\\
&\epsilon+E\bigg(\bigg(Y^W(\tau)-\frac{Z({\tau})}{M(\tau)}\bigg)^{+}\bigg)=\epsilon+
 E\bigg(\lim_{n\rightarrow\infty}\bigg(Y^{\xi,n}({\tau_n})-\frac{
Z^{(n)}({\frac{\tau_nT}{n}})}{M^{(n)}({\frac{\tau_nT}{n}})}\bigg)^{+}\bigg)\nonumber\\
&=\epsilon+ \lim_{n\rightarrow\infty}E[(Y^{\xi,n}({\tau_n})-V^{\pi_n}({\tau_n}))^{+}]
\leq \epsilon+\lim_{n\rightarrow\infty} R_n(x).\nonumber
\end{eqnarray}
Since $\epsilon>0$ was arbitrary we conclude that $R(x)\leq \lim_{n\rightarrow\infty} R_n(x)$.

Next, we show that $R(x)\geq\lim_{n\rightarrow\infty}R_n(x)$.
Choose $\epsilon>0$. From Lemma \ref{lem3.3-} it follows that there
exists $\psi\in C((\mathbb{D}([0,T];\mathbb{R}^d),\mathcal{S}))$ such that the
stochastic process $H(t):=E^W(\psi_i(S^W)|\mathcal{F}^{S^W}_t)$,
$t\leq{T}$ satisfies $H(0)<x$ and
\begin{equation}\label{4.18}
R(x)>\sup_{\tau\in\mathcal{T}^{S^W}_{[0,T]}}E^W\bigg(\bigg(Y^W({\tau})-\frac{H({\tau})}{M({\tau})}\bigg)^{+}\bigg)-\epsilon.
\end{equation}
For any $n$ define the stochastic process
$H^{(n)}(t)=E^{\xi}(\psi(S^{n,\xi})|\mathcal{F}^{S^{n,\xi}}_t)$,
$t\leq{T}$. From Lemma \ref{lem3.4} we obtain
\begin{equation}\label{4.20}
\begin{split}
(S^{\xi,n},H^{(n)})\Rightarrow (S^W,H)  \ \mbox{on} \ \mbox{the} \ \mbox{space} \ (\mathbb{D}([0,T];\mathbb{R}^d),\mathcal{S})\times
(\mathbb{D}([0,T];\mathbb{R}),\mathcal{S}).
\end{split}
\end{equation}
Since the process $H$ is continuous then
$\lim_{n\rightarrow\infty}H^{(n)}(0)=H(0)$. Thus, we will assume that
$n$ is sufficiently large such that $H^{(n)}(0)\leq{x}$. Observe that
the process $\frac{H^{(n)}({kT/n})}{M^{(n)}({kT/n})}$,
$0\leq{k}\leq{n}$ is a martingale with respect to
$\tilde{P}^{\xi}_n$ and the filtration $\{\mathcal{F}^\xi_k\}_{k=0}^n$, thus (since the multinomial markets are complete) there exists
$\pi'_n\in\mathcal{A}^{\xi,n}(x)$ such that $V^{\pi'_n}(k)=\frac{H^{(n)}({kT/n})}{M^{(n)}({kT/n})}$, $k\leq n$.
We obtain that for any $n$
there exists a stopping time
$\sigma_n\in\mathcal{T}^\xi_n$ which satisfies
\begin{eqnarray}\label{4.21}
&E^{\xi}
\bigg(\bigg(Y^{\xi,n}({\sigma_n})-\frac{Z^{(n)}(\frac{\sigma_nT}{n})}{M^{(n)}(\frac{\sigma_nT}{n})}\bigg)^{+}\bigg)>
\sup_{\tau\in\mathcal{T}^\xi_n} E^{\xi}
\bigg(\bigg(Y^{\xi,n}({\tau})-\frac{Z^{(n)}(\frac{\tau T}{n})}{M^{(n)}(\frac{\tau T}{n})}\bigg)^{+}\bigg)\\
&-\frac{1}{n}\geq R_n(\pi'_n)-\frac{1}{n}\geq
R_n(x)-\frac{1}{n}.\nonumber
\end{eqnarray}
From (\ref{3.20---}) and (\ref{4.20}) the sequence
$(S^{\xi,n},H^{(n)},M^{(n)},\sigma_nT/n)$ is tight on the space $(\mathbb{D}([0,T];\mathbb{R}^d),\mathcal{S})\times
(\mathbb{D}([0,T];\mathbb{R}),\mathcal{S})\times(\mathbb{D}([0,T];\mathbb{R}),\mathcal{S})\times [0,T]$. Thus there exists
a subsequence such that
$(S^{\xi,n},H^{(n)},H^{(n)},\sigma_nT/n)\Rightarrow (S^W,H,M,\nu)$ for some random variable $\nu\leq T$.
From the Skorohod representation theorem we can assume that there
exists a probability space $(\Omega,\mathcal{F},P)$ on which
\begin{equation}\label{4.22}
(S^{\xi,n},H^{(n)},M^{(n)},\sigma_nT/n)\rightarrow (S^W,Z,M,\nu) \ \
\mbox{a.s.}
\end{equation}
on the space $(\mathbb{D}([0,T];\mathbb{R}^d),\mathcal{S})\times
(\mathbb{D}([0,T];\mathbb{R}),\mathcal{S})\times(\mathbb{D}([0,T];\mathbb{R}),\mathcal{S})\times [0,T]$. Observe that the joint distribution of $(S^W,Z,M)$ in (\ref{4.22})
remains as the original one. From Lemma 3.3 in \cite{D} it follows that for any $t\leq T$,
$\{\nu\leq{t}\}$ and $\mathcal{F}^{S^W}_T$ are
conditionally independent given $\mathcal{F}^{S^W}_t$, and for any
uniformly integrable $c\grave{a}dl\grave{a}g$ stochastic process
${\{\Phi(t)\}}_{t=0}^T$ adapted to the filtration
$\mathcal{F}^{S^W}_{[0,T]}$
\begin{equation}\label{4.23}
E\Phi({\nu})\leq\sup_{\tau\in\mathcal{T}^S_{[0,T]}}E\Phi({\tau}).
\end{equation}
Finally, by using (\ref{4.23}) for the process $\Phi(t):=(Y^W({t})-\frac{H({t})}{M({t})})^{+}$,
(\ref{4.18}) and
(\ref{4.21})--(\ref{4.22}) we obtain
\begin{eqnarray*}
&\lim_{n\rightarrow\infty}R_n(x)\leq \lim_{n\rightarrow\infty}E^{\xi}
\bigg(\bigg(Y^{\xi,n}({\sigma_n})-\frac{Z^{(n)}(\frac{\sigma_n T}{n})}{M^{(n)}(\frac{\sigma_n T}{n})}\bigg)^{+}\bigg)=\\
&E\bigg(\lim_{n\rightarrow\infty}\bigg(Y^{\xi,n}({\sigma_n})-\frac{Z^{(n)}(\frac{\sigma_n T}{n})}{M^{(n)}(\frac{\sigma_n T}{n})}
\bigg)^{+}\bigg)=
E^W\bigg(\bigg(Y^W({\nu})-\frac{H({\nu})}{M({\nu})}\bigg)^{+}\bigg)< R(x)+\epsilon
\end{eqnarray*}
and the proof is completed.
\begin{rem}
An interesting question is whether Theorem \ref{thm2.1} is valid for
game options which were introduced in \cite{K}. Let
$F,G:[0,T]\times(\mathbb{D}([0,T];\mathbb{R}^d),\mathcal{S})\rightarrow\mathbb{R}_{+}$
such that $F\leq G$ satisfy the assumptions after (\ref{2.2}). Set
\begin{eqnarray}\label{4.30}
&H^W(t,s)=G(t,S^W)\mathbb{I}_{t<s}+F(s,S^W)\mathbb{I}_{s\leq t}, \ \ t,s\in [0,T] \ \mbox{and}\\
&H^{\xi,n}(k,l)=G(\frac{kT}{n},S^{\xi,n})\mathbb{I}_{k<l}+F(\frac{lT}{n},S^{\xi,n})\mathbb{I}_{l\leq k}, \ \ n\in\mathbb{N},\ 0\leq k,l\leq n.\nonumber
\end{eqnarray}
The terms $H^W(t,s)$ and $H^{\xi,n}(k,l)$ are the payoff functions for the BS model
and the $n$--step multinomial model, respectively.
For game options the shortfall risk is defined by (see \cite{DK1})
\begin{eqnarray}\label{4.31}
&R^{(g)}(x)=\inf_{\pi\in\mathcal{A}^W(x)}\inf_{\sigma\in\mathcal{T}^W_{[0,T]}}\sup_{\tau\in\mathcal{T}^W_{[0,T]}}
E^W[(H^W(\sigma,\tau)-V^{\pi}({\sigma\wedge\tau}))^{+}] \\
& \mbox{and} \
R^{(g)}_n(x)=\inf_{\pi\in\mathcal{A}^{\xi,n}(x)}\min_{\sigma\in\mathcal{T}^{\xi}_n}\max_{\tau\in\mathcal{T}^{\xi}_n}E^\xi
[(H^{\xi,n}(\sigma,\tau)-V^{\pi}({\sigma\wedge\tau}))^{+}].\nonumber
\end{eqnarray}
The question is whether the equality $R^{(g)}(x)=\lim_{n\rightarrow\infty}R^{(g)}_n(x)$ holds true.
Following the proof above it can be shown that
$R^{(g)}(x)\geq\limsup_{n\rightarrow\infty}R^{(g)}_n(x)$. The inequality
$R^{(g)}(x)\leq\liminf_{n\rightarrow\infty}R^{(g)}_n(x)$ is more difficult to
prove because of the additional $\inf$ (in formula (\ref{4.31}))
which destroys the convexity that was used in (\ref{4.17-}) (by
applying Jensen's inequality). At present it is not clear whether
the weak convergence approach can be applied here.
\end{rem}

\section{Analysis of the multinomial models}\label{sec5}\setcounter{equation}{0}
In this section we provide a dynamical programming algorithm for the shortfall risks and the corresponding optimal portfolios
in the multinomial models. Similar analysis was done in \cite{DK1}
for game options in multinomial markets with one risky asset.
\begin{defn}\label{defn5.1}
A function $\psi:{\mathbb{R}_{+}}\rightarrow{\mathbb{R}_{+}}$ is
 a piecewise linear function vanishing at $\infty$ if there exists a natural
 number $n$, such that
\begin{equation}\label{5.1}
\psi(y)=\sum_{i=1}^n \mathbb{I}_{[a_i,a_{i+1})}(c_iy+d_i)
\end{equation}
where $c_1,...,c_n,d_1,...,d_n\in\mathbb{R}$ and $a_1<a_2<...<a_{n+1}<\infty$.
\end{defn}

Let $J=\{v^{(1)},...,v^{(d+1)}\}\subset \mathbb{R}^d$ such that
\begin{equation}\label{5.1+}
\begin{split}
span\{v^{(1)},...,v^{(d+1)}\}=\mathbb{R}^d \ \mbox{and} \ \exists p_1,...,p_{d+1}>0,  \
 \sum_{i=1}^{d+1} p_i v^{(i)}=0.
\end{split}
\end{equation}
Define the set $K_J=\{u\in\mathbb{R}^{d}|\langle u,v^{(i)}\rangle\geq -1, \ i=1,...,d+1\}$.
Observe that $K_J$ is a compact convex set.
\begin{lem}\label{lem5.2}
Let $\psi_1,...,\psi_{d+1}:{\mathbb{R}_{+}}\rightarrow{\mathbb{R}_{+}}$ be
continuous, non increasing and piecewise linear functions vanishing at
$\infty$. Define $\psi:{\mathbb{R}_{+}}\rightarrow{\mathbb{R}_{+}}$
by
\begin{equation}\label{5.3}
\psi(y)=\min_{u\in K_J}\sum_{i=1}^{d+1}\psi_i(y(1+\langle u,v^{(i)}\rangle)).
\end{equation}
Then $\psi$ is continuous, non increasing and piecewise
linear function vanishing at $\infty$.
\end{lem}
\begin{proof}
Clearly $\psi$ is a non increasing function. There exists a natural number $n$ such that
\begin{equation}\label{5.5}
\psi_i(y)=\sum_{j=1}^n \mathbb{I}_{[a_j,a_{j+1})}(c^{(i)}_jy+d^{(i)}_j),
\ i=1,...,d+1
\end{equation}
where $c^{(i)}_j,d^{(i)}_j\in{\mathbb{R}}$ and
$0=a_1<a_2<...<a_{n+1}<\infty$. Denote
$I_k=[a_k,a_{k+1})$, $k=1,...,n$ and $I_{n+1}=[a_{n+1},\infty)$.
Set
$\lambda_i=1+\sup_{u\in K}\langle u,v^{(i)}\rangle$, $i\leq d+1$.
Notice that for any $y_1,y_2\in\mathbb{R}_{+}$
\begin{eqnarray}\label{5.5+}
&|\psi(y_1)-\psi(y_2)|\leq\sum_{i=1}^{d+1}\sup_{u\in K}|\psi_i(y_1(1+\langle u,v^{(i)}\rangle))
-\psi_i(y_2(1+\langle u,v^{(i)}\rangle))| \\
&\leq|y_1-y_2|\sum_{i=1}^{d+1}\lambda_i \max_{1\leq j\leq n}|c^{(i)}_j|.
\nonumber
\end{eqnarray}
Thus $\psi$ is a continuous function.
Next, we prove that $\psi$ is a piecewise linear function. Fix $y>0$ and
introduce the set
 $L_y=\{\frac{a_1}{y}-1,....,\frac{a_{n+1}}{y}-1\}$.
For any $1\leq\alpha\leq d+1$ and
 $\beta\in \{1,...,n+1\}^{d+1}$ define the sets
 $L^{(y)}_{\alpha}=
\{w\in\mathbb{R}^{d}|\langle v^{(i)},w\rangle\in L_y, i\in \{1,...,d+1\}\setminus\{\alpha\}\}$ and
$K^{(y)}_{\beta}=\{u\in\mathbb{R}^d|y(1+\langle v^{(j)},u\rangle)\in I_{\beta_j}, \ \forall{j\leq d+1} \}$.
Set $L^{(y)}=\bigcup_{\alpha=1}^{d+1}L^{(y)}_{\alpha}$.
There exists a finite sequence of real numbers $c_1,...,c_m,e_1,...,e_m$ (which does not depend on $y$)
such that any $v\in L^{(y)}$ is of the form
$v=(c_{k_1}+\frac{e_{r_1}}{y},...,c_{k_d}+\frac{e_{r_d}}{y})$,
$k_1,...,k_d,r_1,...,r_d\in\{1,...,m\}$.
Notice that for any $\beta$, $K^{(y)}_\beta\subset K_J$ is a compact convex set.
Furthermore, the extreme points of $K^{(y)}_{\beta}$
are in $L^{(y)}$. For each $\beta\in \{1,...,n+1\}^{d+1}$ the function
$\psi^{(y)}:K^{(y)}_\beta\rightarrow\mathbb{R}_{+}$ which given by
$\psi^{(y)}(u)=\sum_{i=1}^{d+1}
\psi_i(y(1+\langle v^{(i)},u\rangle))$, is a convex function.
Since $\bigcup_{\beta\in\{1,...,n+1\}^{d+1}}K^{(y)}_{\beta}=K_J$,
we obtain
\begin{eqnarray}\label{5.6}
&\psi(y)=\min_{\beta\in \{1,...,n+1\}^{d+1}}\min_{u\in K^{(y)}_{\beta}}\psi^{(y)}(u)=\\
&\min_{\beta\in \{1,...,n+1\}^{d+1}}\min_{u\in K^{(y)}_{\beta}\bigcap L^{(y)}}\psi^{(y)}(u)=
\min_{u\in K_J\bigcap L^{(y)}}\psi^{(y)}(u).\nonumber
\end{eqnarray}
Thus
there exists a finite sequence of real numbers $f_1,...,f_{\tilde{m}},g_1,...,g_{\tilde{m}}$
such that for any $y>0$,
\begin{equation}\label{5.7}
\psi(y)=f_i y+g_i
\end{equation}
for some $i$ (which depends on $y$). This together with
the inequality $\psi(y)\leq\sum_{i=1}^{d+1}\psi_i(y)$
and the fact that $\psi$ is a continuous function
gives that $\psi$ is a piecewise linear function vanishing at
$\infty$.
\end{proof}

Next, fix $n$ and consider the $n$--step multinomial model.
For any $\pi\in\mathcal{A}^{\xi,n}$ define a sequence of
random variables ${\{U^\pi(k)\}}_{k=0}^n$ by
\begin{eqnarray}\label{5.9}
&U^\pi(n)=({Y}^{\xi,n}(n)-{V}^\pi(n))^+, \ \mbox{and} \ \mbox{for} \ k<n\\
&U^\pi(k)=\max(E^{\xi}(U^{\pi}({k+1})
|\mathcal{F}^{\xi}_k),({Y}^{\xi,n}(k)-V^\pi(k))^{+}).\nonumber
\end{eqnarray}
Applying standard results for optimal stopping (see \cite{PS})
for the process $({Y}^{\xi,n}(k)-{V}^{\pi}(k))^{+}$, $k=0,1,...,n$
we obtain
\begin{equation}\label{5.10}
U^\pi(0)=
\max_{\tau\in{\mathcal{T}^\xi_{n}}}E^{\xi}[(Y^{\xi,n}({\tau})-{V}^{\pi}({\tau}))^+]=R_n(\pi).
\end{equation}
Set,
\begin{eqnarray}\label{5.10+}
&w^{(i)}=\sqrt{d+1}(A_{i1},...,A_{id}), \ w^{n,i}=\frac{T}{n}b+\sqrt{\frac{T}{n}}w^{(i)}\sigma^{*}, \ \  i\leq d+1\\
&J=\{w^{(1)},...,w^{(d+1)}\} \ \mbox{and} \ J_n=\{w^{n,1},...,w^{n,d+1}\}\nonumber
\end{eqnarray}
where the matrix $A$ and the vector $b$ were introduced in Section 2.
\begin{defn}\label{defn5.2}
Let $0\leq k<n$ and $X$ be a nonnegative $\mathcal{F}^{\xi}_k$--measurable
random variable. Define the set
\begin{eqnarray}\label{5.10++}
&\mathcal{A}^{(n)}_k(X)=\bigg\{Y|Y=X\bigg(1+\bigg\langle \rho,\frac{T}{n}b+
{\sqrt\frac{T}{n}}\xi^{(k+1)}\sigma^{*} \bigg\rangle\bigg), \\
&\rho:\Omega_{\xi}\rightarrow K_{J_n} \ \mbox{is} \
\mathcal{F}^\xi_k-\mbox{measurable} \bigg\}.\nonumber
\end{eqnarray}
\end{defn}

Notice that if
$V^\pi(k)=X$ and $V^\pi(k+1)=Y$ for some $\pi=(\gamma(1),...,\gamma(n))\in\mathcal{A}^{\xi,n}$ and $k<n$
then from (\ref{2.12}) and (\ref{2.15}),
$Y=X(1+\langle \rho,\frac{T}{n}b+
\sqrt{\frac{T}{n}}\xi^{(k+1)}\sigma^{*}\rangle)$
where  $\rho=\frac{\mathbb{I}_{X>0}}{X}(\gamma_1(k+1)S^{\xi,n}_1(\frac{(k+1)T}{n}),...,
\gamma_d(k+1)S^{\xi,n}_d(\frac{(k+1)T}{n}))$.
Clearly, if $X=0$ then ($\pi$ is \textit{admissible}) $Y=0$. Since we require
$Y\geq 0$ to be satisfied for all possible values of $\xi^{(k+1)}$ then
in view of independency of $\rho$ and $\xi^{(k+1)}$ we conclude that
$\mathcal{A}^{(n)}_k(X)$ is the set of all possible portfolio values at time
$k+1$ provided the portfolio value at time $k$ is $X$.

For any $0\leq k\leq n$ let $\phi^{(n)}_k:J^k\rightarrow\mathbb{R}_{+}$ such that
\begin{equation}\label{5.11-}
\phi^{(n)}_k(\xi^{(1)},...,\xi^{(k)})=Y^{\xi,n}(k).
\end{equation}
Define a sequence of functions
$H^{(n)}_k:\mathbb{R}_{+}\times J^k\rightarrow \mathbb{R}_{+},\, k=0,1,...,
n$ by the following backward relations. For any $u^{(1)},...,u^{(n)}\in J$ and $y\in\mathbb{R}_{+}$
\begin{eqnarray}\label{5.11}
&H^{(n)}_n(y,u^{(1)},...,u^{(n)})=(\phi^{(n)}_n(u^{(1)},...,u^{(n)})-y)^+ \  \ \mbox{and}\\
&H^{(n)}_k(y,u^{(1)},...,u^{(k)})=\max\bigg(\phi^{(n)}_n(u^{(1)},...,u^{(k)})-y)^{+}, \ \
\frac{1}{d+1} \inf_{u\in {K_{J_n}}}\nonumber\\
&\sum_{i=1}^{d+1} H^{(n)}_{k+1}(y(1+\langle u,\frac{T}{n}b+\sqrt{\frac{T}{n}}w^{(i)}\sigma^{*}\rangle),
u^{(1)},...,u^{(k)},w^{(i)})\bigg) \ \mbox{for}  \ k<n.\nonumber
\end{eqnarray}
Observe that $J_n$ (for sufficiently large $n$) satisfies
(\ref{5.1+}). Thus from Lemma \ref{lem5.2} it follows (by backward induction) that for any $k\leq n$ and
$u^{(1)},...,u^{(k)}\in J$,
$H^{(n)}_k(\cdot,u^{(1)},...,u^{(k)})$ is continuous, non increasing and piecewise
linear function vanishing at $\infty$. These facts allow us to define the functions
$\{h^{(n)}_k:\mathbb{R}_{+}\times J^k\rightarrow K_{J_n}\}_{k=0}^{n-1}$
by
\begin{eqnarray}\label{5.11+}
h^{(n)}_k(y,u^{(1)},...,u^{(k)})=argmin_{u\in K_{J_n}}
\sum_{i=1}^{d+1}H^{(n)}_{k+1}(y(1+\langle u,\\
\frac{T}{n}b+\sqrt{\frac{T}{n}}w^{(i)}\sigma^{*}\rangle),
u^{(1)},...,u^{(k)},w^{(i)}).\nonumber
\end{eqnarray}
Namely,
\begin{eqnarray}\label{5.12}
&\min_{u\in {K_{J_n}}}\sum_{i=1}^{d+1}H^{(n)}_{k+1}(y(1+\langle u,\frac{T}{n}b+\sqrt{\frac{T}{n}}w^{(i)}\sigma^{*}\rangle),
u^{(1)},...,u^{(k)},w^{(i)})=\\
&\sum_{i=1}^{d+1}H^{(n)}_{k+1}(y(1+\langle h^{(n)}_k(y,u^{(1)},...,u^{(k)}),\frac{T}{n}b+\sqrt{\frac{T}{n}}w^{(i)}\sigma^{*}\rangle),
u^{(1)},...,u^{(k)},w^{(i)})\nonumber
\end{eqnarray}
for any $y\in\mathbb{R}_{+}$ and $u^{(1)},...,u^{(k)}\in J$.

Let $x>0$ be an initial capital. Define $\tilde{\pi}=\tilde\pi(n,x)\in\mathcal{A}^{\xi,n}(x)$ by
\begin{eqnarray}\label{5.13}
&V^{\tilde\pi}(0)=x, \ \mbox{and} \ \mbox{for} \ 0\leq k<n\\
& V^{\tilde\pi}(k+1)=
V^{\tilde\pi}(k)(1+\langle h^{(n)}_k(V^{\tilde\pi}(k),\xi^{(1)},...,\xi^{(k)}),\frac{T}{n}b+\sqrt{\frac{T}{n}}\xi^{(k+1)}\sigma^{*}\rangle).
\nonumber
\end{eqnarray}
\begin{thm}
For any $n\in\mathbb{N}$ and $x>0$
\begin{equation}\label{5.14}
R_n(x)=R_n(\tilde\pi(n,x))=H^{(n)}_0(x).
\end{equation}
\end{thm}
\begin{proof}
Fix $n\in\mathbb{N}$ and $x>0$. Let $\pi\in\mathcal{A}^{\xi,n}(x)$ an arbitrary
portfolio. Denote $\tilde{\pi}=\tilde\pi(n,x)$. First we prove by backward induction that
for any $k\leq n$,
\begin{equation}\label{5.15}
H^{(n)}_k(V^{\pi}(k),\xi^{(1)},...,\xi^{(k)})\leq U^{\pi}(k)
\ \mbox{and} \
H^{(n)}_k(V^{\tilde\pi}(k),\xi^{(1)},...,\xi^{(k)})= U^{\tilde\pi}(k).
\end{equation}
For $k=n$,
we obtain from (\ref{5.9}) and (\ref{5.11-})--(\ref{5.11})
that the relations (\ref{5.15}) hold with equality. Suppose that
(\ref{5.15}) holds true for $k+1$ and prove them for $k$. Let
$\rho:\Omega_{\xi}\rightarrow K_{J_n}$ be a
$\mathcal{F}^\xi_k$ measurable random vector such that
$V^\pi(k+1)=V^\pi(k)(1+\langle \rho,\frac{T}{n}b+
{\sqrt\frac{T}{n}}\xi^{(k+1)}\sigma^{*} \rangle)$. From the induction assumption we obtain
\begin{eqnarray}\label{5.16}
&E^{\xi}(U^\pi(k+1)|\mathcal{F}^\xi_k)\geq
E^{\xi}(H^{(n)}_k(V^{\pi}(k+1),\xi^{(1)},...,\xi^{(k)},\xi^{(k+1)})|\mathcal{F}^\xi_k)\\
&=\frac{1}{d+1}
\sum_{i=1}^{d+1} H^{(n)}_{k+1}(V^\pi(k)(1+\langle \rho,\frac{T}{n}b+\sqrt{\frac{T}{n}}w^{(i)}\sigma^{*}\rangle),
\xi^{(1)},...,\xi^{(k)},w^{(i)})\geq\nonumber\\
&\frac{1}{d+1} \inf_{u\in {K_{J_n}}}
\sum_{i=1}^{d+1} H^{(n)}_{k+1}(V^{\pi}(k)(1+\langle u,\frac{T}{n}b+\sqrt{\frac{T}{n}}w^{(i)}\sigma^{*}\rangle),
\xi^{(1)},...,\xi^{(k)},w^{(i)}).\nonumber
\end{eqnarray}
Denote $\tilde{\rho}=h^{(n)}_k(V^{\tilde\pi}(k),\xi^{(1)},...,\xi^{(k)})$.
From (\ref{5.12})--(\ref{5.13}) and the induction assumption it follows
\begin{eqnarray}\label{5.17}
&E^{\xi}(U^{\tilde\pi}(k+1)|\mathcal{F}^\xi_k)=
E^{\xi}(H^{(n)}_k(V^{\tilde\pi}(k+1),\xi^{(1)},...,\xi^{(k)},\xi^{(k+1)})|\mathcal{F}^\xi_k)\\
&=\frac{1}{d+1}
\sum_{i=1}^{d+1} H^{(n)}_{k+1}(V^{\tilde\pi}(k)(1+\langle \tilde{\rho},\frac{T}{n}b+\sqrt{\frac{T}{n}}w^{(i)}\sigma^{*}\rangle),
\xi^{(1)},...,\xi^{(k)},w^{(i)})=\nonumber\\
&\frac{1}{d+1} \inf_{u\in {K_{J_n}}}
\sum_{i=1}^{d+1} H^{(n)}_{k+1}(V^{\tilde\pi}(k)(1+\langle u,\frac{T}{n}b+\sqrt{\frac{T}{n}}w^{(i)}\sigma^{*}\rangle),
\xi^{(1)},...,\xi^{(k)},w^{(i)}).\nonumber
\end{eqnarray}
Combining (\ref{5.9}), (\ref{5.11-})--(\ref{5.11}) and (\ref{5.16})-(\ref{5.17}) we obtain that
(\ref{5.15}) holds true. Next, by using
(\ref{5.15}) for $k=0$ and (\ref{5.10}) it follows that for any $\pi\in\mathcal{A}^{\xi,n}(x)$
\begin{equation*}
R_n(\pi)=U^\pi(0)\geq  H^{(n)}_0(V^\pi(0))\geq H^{(n)}_0(x)=U^{\tilde\pi}(0)=R_n(\tilde\pi).
\end{equation*}
Thus $R_n(x)=R_n(\tilde\pi)=H^{(n)}_0(x)$, as required.
\end{proof}
${}$\\
\begin{Large}
\textbf{Acknowledgments}
\end{Large}
\\
\\
Part of this work was done during my P.hD studies at the Hebrew University.
 I would like to express my deepest gratitude to
my P.hD adviser, Yuri Kifer, for his guidance throughout my graduate studies.
I am also very grateful to Martin Schweizer for helping me to present this work.
During my P.hD studies I was partially supported by ISF grant no. 130/06.

\end{document}